\documentclass[a4paper,10pt]{article}

\usepackage[utf8x]{inputenc}
\usepackage{amsmath}
\usepackage{amsfonts}
\usepackage{amssymb}
\usepackage{graphicx}
\usepackage{color}
\usepackage{amsthm}

\newtheorem{rem}{Remark}

\newtheorem{cor}{Corollary}
\newtheorem{theo}{Theorem}
\newtheorem{conj}{Conjecture}
\newtheorem{claim}{Claim}

\def\H{\mathcal{H}}
\def\K{\mathcal{K}}
\def\G{\mathcal{G}}

\title{Packing of Rigid Spanning Subgraphs and Spanning Trees}

\author{Joseph Cheriyan \\ Olivier Durand de Gevigney \\ Zolt\'an Szigeti}

\begin{document}
\maketitle

\begin{abstract} 
We prove that every $(6k+2\ell,2k)$-connected simple graph contains $k$ rigid and $\ell$ connected edge-disjoint spanning subgraphs. This implies a theorem of Jackson and Jord\'an \cite{Jackson_Jordan2009} and a theorem of Jord\'an \cite{Jordan_2005} on packing of rigid spanning subgraphs. Both these results are generalizations of the classical result of Lov\'asz and Yemini \cite{Lovasz_Yemini} saying that every $6$-connected graph is rigid for which our approach provides a transparent proof.
Our result also gives two improved upper bounds on the connectivity of graphs that have interesting properties: (1) every $8$-connected graph packs a spanning tree and a $2$-connected spanning subgraph; (2) every $14$-connected graph has a $2$-connected orientation.
\end{abstract}

\section{Definitions}

Let $G=(V,E)$ be a graph. 
We will use the following connectivity concepts. 
$G$ is called {\bf connected} if for every  pair $u,v$ of vertices there is a path from $u$ to $v$ in $G$.
$G$ is called {\bf $k$-edge-connected} if $G-F$ is connected for all $F\subseteq E$ with $|F|\leq k-1$.
$G$ is called {\bf $k$-connected} if $|V|>k$ and $G-X$ is  connected for all $X\subset V$ with $|X|\leq k-1$.
For a pair of positive integers $(p,q)$, $G$ is called {\boldmath $(p,q)$}-\textbf{connected} if $G-X$ is $(p-q|X|)$-edge-connected for all $X \subset V$.
By Menger theorem, $G$ is $(p,q)$-connected if and only if
for every pair of disjoint subsets $X,Y$ of $V$ such that $Y\neq\emptyset , X \cup Y \neq V$, 
\begin{equation}
d_{G-X}(Y)\geq p -q|X|.
\end{equation}
For a better understanding we mention that $G$ is $(6,2)$-connected if $G$ is $6$-edge-connected, $G-v$  is $4$-edge-connected for all $v\in V$ and $G-u-v$  is $2$-edge-connected for all $u,v\in V.$ 
It follows from the definitions that $k$-edge-connectivity is equivalent to $(k,k)$-connectivity. Moreover, since loops and parallel edges do not play any role in vertex connectivity, every $k$-connected graph contains a $(k,1)$-connected simple spanning subgraph.
Note also that $(k,1)$-connectivity implies $(k,q)$-connectivity for all $q \geq 1$. (Remark that this connectivity concept is (very slightly) different from the one introduced by Kaneko and Ota \cite{Kaneko_Ota2000} since $p$ is not required to be a multiple of $q$.)
\medskip

Let $D=(V,A)$ be a directed graph. $D$  is called {\bf strongly connected} if for every ordered pair $(u,v)\in V\times V$ of vertices there is a directed path from $u$ to $v$ in $D$.  $D$ is called {\bf $k$-arc-connected} if $G-F$
is  strongly connected for all $F\subseteq A$ with $|F|\leq k-1$. $D$ is called {\bf $k$-connected} if $|V|>k$ and $G-X$ is  strongly connected for all $X\subset V$ with $|X|\leq k-1$.
\medskip

For a set $X$ of vertices and  a set $F$ of edges, denote {\boldmath $G_F$} the subgraph of $G$ on vertex set $V$ and edge set $F,$ that is $G_F=(V,F)$ and {\boldmath $E(X)$} the set of edges of $G$ induced by $X$. 
Denote {\boldmath $\mathcal{R}(G)$} the \textbf{rigidity matroid} of $G$ on ground-set $E$ with rank function {\boldmath  $r_{\mathcal{R}}$} (for a definition we refer the reader to \cite{Lovasz_Yemini}). For $F \subseteq E$, by a theorem of Lov\'asz and Yemini \cite{Lovasz_Yemini},
\begin{equation} \label{rr}
r_{\mathcal{R}}(F) = \min \sum_{X \in \H} (2|X| - 3),
\end{equation}
where the minimum is taken over all collections $\H$ of subsets of $V$ such that $\{E(X) \cap F, X \in \H\}$ partitions $F$.
\begin{rem}\label{Xconnected}
If $\H$ achieves the minimum in (\ref{rr}), then each $X \in \H$ induces a connected subgraph of $G_F$.
\end{rem}
We will say that $G$ is {\bf rigid} if $r_{\mathcal{R}}(E) = 2|V|-3.$
\medskip

\section{Results}

Lov\'asz and Yemini \cite{Lovasz_Yemini} proved the following sufficient condition for a graph to be rigid.
\begin{theo}[Lov\'asz and Yemini \cite{Lovasz_Yemini}] \label{Lo_Ye}
 Every $6$-connected graph is rigid.
\end{theo}

Jackson and Jord\'an \cite{Jackson_Jordan2009} proved a sharpenning of Theorem \ref{Lo_Ye}.
\begin{theo}[Jackson and Jord\'an \cite{Jackson_Jordan2009}] \label{Ja_Jo}
 Every $(6,2)$-connected simple graph is rigid.
\end{theo}

Jord\'an \cite{Jordan_2005} generalized Theorem \ref{Lo_Ye} and gave a sufficient condition for the existence of a packing of rigid spanning subgraphs.
\begin{theo}[Jord\'an \cite{Jordan_2005}] \label{Jordan_main}
 Let $k \geq 1$ be an integer. Every $6k$-connected graph contains $k$ edge-disjoint rigid spanning subgraphs.
\end{theo}

The main result of this paper contains a common generalization of Theorems \ref{Ja_Jo} and \ref{Jordan_main}. It provides a sufficient condition to have a packing of rigid spanning subgraphs and spanning trees.
\begin{theo} \label{main_th}
 Let $k \geq 1$ and $\ell \geq 0$ be integers. Every $(6k+2\ell,2k)$-connected simple graph contains $k$ rigid spanning subgraphs and $\ell$ spanning trees pairwise edge-disjoint.
\end{theo}
Note that in Theorem \ref{Ja_Jo}, the connectivity condition is the best possible since there exist non-rigid $(5,2)$-connected graphs (see \cite{Lovasz_Yemini}) and non-rigid $(6,3)$-connected graphs, for an example see Figure \ref{fig:ex}.

\begin{figure}[h]
 \centering
 \includegraphics[width=5cm]{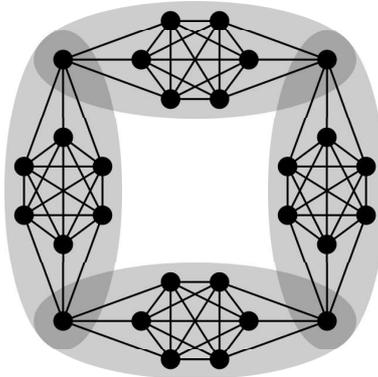}

 \caption{A $(6,3)$-connected non-rigid graph $G=(V,E)$. The collection $\H$ of the four grey vertex-sets partitions $E$. Hence, by (\ref{rr}), $\mathcal{R}_G(E) \leq \sum_{X \in \H}(2|X|-3) = 4(2\times8-3) = 52 < 53 = 2\times28 - 3 = 2|V|-3$. Thus $G$ is not rigid. The reader can easily check that $G$ is $(6,3)$-connected.}
 \label{fig:ex}
 % dessin.pdf: 148x148 pixel, 72dpi, 5.22x5.22 cm, bb=0 0 148 148
\end{figure}
Let us see some corollaries of the previous results.
Theorem \ref{main_th} applied for $k=1$ and $\ell=0$ provides Theorem \ref{Ja_Jo}.
Since $6k$-connectivity implies $(6k,2k)$-connectivity of a simple spanning subgraph, Theorem \ref{main_th} implies Theorem \ref{Jordan_main}.\\

One can easily derive from the rank function of $\mathcal{R}(G)$ that rigid graphs with at least $3$ vertices are $2$-connected (see Lemma 2.6 in \cite{Jackson_Jordan2005}). Thus, Theorem \ref{main_th} gives the following corollary.
\begin{cor} \label{main_cor}
 Let $k \geq 1$ and $\ell \geq 0$ be integers. Every $(6k+2\ell,2k)$-connected simple graph contains $k$ $2$-connected and $\ell$ connected edge-disjoint spanning subgraphs.
\end{cor}
Corollary \ref{main_cor} allows us to improve two results of Jord\'an. The first one deals with the following conjecture of Kriesell, see in \cite{Jordan_2005}.
\begin{conj}[Kriesell]
For every positive integer $\lambda$ there exists a (smallest) $f(\lambda)$ such that every $f(\lambda)$-connected graph $G$ contains a spanning tree $T$ for which $G − E(T)$ is $\lambda$-connected.
\end{conj}
As Jord\'an pointed out in \cite{Jordan_2005}, Theorem \ref{Jordan_main} answers this conjecture for $\lambda=2$ by showing that $f(2)\leq 12$. Corollary \ref{main_cor} applied for $k=1$ and $\ell=1$ directly implies that $f(2) \leq 8$.
\begin{cor}
 Every $8$-connected graph $G$ contains a spanning tree $T$ such that $G-E(T)$ is $2$-connected.
\end{cor}

The other improvement deals with the following conjecture of Thomassen \cite{Thomassen_1989}.
\begin{conj}[Thomassen \cite{Thomassen_1989}]
For every positive integer $\lambda$ there exists a (smallest) $g(\lambda)$ such that every $g(\lambda)$-connected graph $G$ has a $\lambda$-connected orientation.
\end{conj}
By applying Theorem  \ref{Jordan_main} and an orientation result of Berg and Jord\'an \cite {Berg_Jordan},
Jord\'an proved in \cite{Jordan_2005} the conjecture for $\lambda=2$ by showing that $g(2)\leq 18$. Corollary \ref{main_cor} allows us to prove a general result that implies $g(2) \leq 14$. For this purpose, we use a result of Kir\'aly and Szigeti \cite{KiraSzig}.
\begin{theo}[Kir\'aly and Szigeti \cite{KiraSzig}] \label{KSz}
 An Eulerian graph $G=(V,E)$ has an Eulerian orientation $D$ such that $D-v$ is $k$-arc-connected for all $v \in V$ if and only if $G-v$ is $2k$-edge-connected for all $v \in V$.
\end{theo}
Corollary \ref{main_cor} and Theorem \ref{KSz} imply the following corollary which gives the claimed bound for $k=1$.
\begin{cor}
 Every simple $(12k+2,2k)$-connected graph $G$ has an orientation $D$ such that $D-v$ is $k$-arc-connected for all $v \in V$.
\end{cor}
\begin{proof} Let $G=(V,E)$ be a simple $(12k+2,2k)$-connected graph. By Theorem \ref{KSz} it suffices to prove that $G$ contains an Eulerian spanning subgraph $H$ such that $H−v$ is $2k$-edge-connected for all $v \in V$. By Corollary \ref{main_cor}, $G$ contains $2k$ $2$-connected spanning subgraphs $H_i=(V,E_i), i=1,\dots,2k$ and a spanning tree $F$ pairwise edge-disjoint. Define $H'=(V, \cup_{i=1}^{2k} E_i)$. For all $i=1,\dots,2k$, since $H_i$ is $2$-connected, $H_i-v$ is connected; hence $H'-v$ is $2k$-edge-connected for all $v \in V$. Denote $T$ the set of vertices of odd degree in $H'$.
We say that $F'$ is a {\boldmath $T$}{\bf -join} if the set of odd degree vertices of $G_{F'}$ coincides with $T.$
It is well-known  that the connected graph $F$ contains a $T$-join. Thus adding the edges of this $T$-join to $H'$ provides the required spanning subgraph of $G$.\end{proof}
\medskip

Finally we mention that the following conjecture of Frank, that would give a necessary and sufficient condition for a graph to have a $2$-connected orientation, would imply that $g(2)\leq 4.$

\begin{conj}[Frank \cite{Frank_1995}]
A graph has a $2$-connected orientation if and only if it is $(4,2)$-connected.
\end{conj}

\section{Proofs}

To prove Theorem \ref{main_th} we need to introduce two other matroids on the edge set $E$ of $G$.
Denote {\boldmath $\mathcal{C}(G)$} the \textbf{circuit matroid} of $G$ on ground-set $E$ with rank function {\boldmath $r_{\mathcal{C}}$} given by (\ref{rc}). Let $n$ be the number of vertices in $G$, that is {\boldmath $n=|V|$}. For $F \subseteq E$, denote {\boldmath $c(G_F)$} the number of connected components of $G_F$, it is well known that,
\begin{equation} \label{rc}
r_{\mathcal{C}}(F) = n - c(G_F).
\end{equation}

To have $k$ rigid spanning subgraphs and $\ell$ spanning trees pairwise edge-disjoint in $G$, we must find $k$ basis in $\mathcal{R}(G)$ and $\ell$ basis in $\mathcal{C}(G)$ pairwise disjoint. To do that we will need the following matroid.
For $k \geq 1$ and $\ell \geq 0$, define {\boldmath $\mathcal{M}_{k,\ell}(G)$} as the matroid on ground-set $E$, obtained by taking the matroid union of $k$ copies of the rigidity matroid $\mathcal{R}(G)$ and $\ell$ copies of the circuit matroid $\mathcal{C}(G)$. Let {\boldmath $r_{\mathcal{M}_{k,\ell}}$} be the rank function of $\mathcal{M}_{k,\ell}(G)$. 
By a theorem of Edmonds \cite{Edmonds_1968}, for the rank of matroid unions, 
\begin{equation} \label{rs}
r_{\mathcal{M}_{k,\ell}}(E) = \min_{F \subseteq E} kr_{\mathcal{R}}(F) +\ell r_{\mathcal{C}}(F) + |E \setminus F|.
\end{equation}

In \cite{Jordan_2005}, Jord\'an used the matroid $\mathcal{M}_{k,0}(G)$ to prove Theorem \ref{Jordan_main} and pointed out that using $\mathcal{M}_{k,\ell}(G)$ one could prove a theorem on packing of rigid spanning subgraphs and spanning trees. We tried to fulfill this gap by following the proof of \cite{Jordan_2005} but we failed. To achieve this aim we had to find a new proof technique.
Let us first demonstrate this technique by giving a transparent proof for Theorems \ref{Lo_Ye} and \ref{Ja_Jo}.

\begin{proof}[Proof of Theorem \ref{Lo_Ye}]  By (\ref{rr}), there exists a collection $\G$ of subsets of $V$ such that $\{E(X), X \in \G\}$ partitions $E$ and $r_{\mathcal{R}}(E) = \sum_{X \in \G} (2|X| - 3)$. If $V \in \G$ then $r_{\mathcal{R}}(E) \geq 2|V| - 3$ hence $G$ is rigid. So in the following we may assume that $V \notin \G$.

Let $\H=\{X\in \G : |X|\geq 3\}$ and $F=\bigcup _{X\in \H}E(X)$. We define, for $X \in \H$, the border of $X$ as $X_B = X \cap (\cup_{Y \in \H-X} Y)$ and the proper part of $X$ as $X_I = X \setminus X_B$ and $\H' = \{X \in \H \textrm{ : } X_I \neq \emptyset\}$.
\smallskip

Since every edge of $F$ is induced by an element of $\H$, for $X \in \H'$, by definition of $X_I$, no edge of $F$ contributes to $d_{G-X_B}(X_I)$; and for a vertex $v\in V-V(\H)$, no edge of $F$ contributes to $d_G(v)$.
Thus, since for $X \in \H'$, $X_I \neq \emptyset$ and $X_I \cup X_B = X \neq V$, by $6$-connectivity of $G$, we have
$|E\setminus F|\geq\frac{1}{2}(\sum_{X \in \H'} d_{G-X_B}(X_I) + \sum_{v \in  V-V(\H)}d_G(v))\geq\frac{1}{2}(\sum_{X \in \H'} (6-|X_B|)+ \sum_{v \in  V-V(\H)}6) \geq 3|\H'|-\sum_{X \in \H'} |X_B|+3(|V|-|V(\H)|).$
\medskip

Since for  $X\in\H\setminus \H'$, $|X_B|=|X|\geq 3,$ we have $\sum_{X \in \H} (2|X| - 3)=\sum_{X \in \H} 2|X| -3|\H|+ 3|\H'| - 3|\H'|\geq\sum_{X \in \H} 2|X| -\sum_{X \in \H\setminus \H'} |X_B| - 3|\H'|.$
\medskip

Since $G$ is simple, by Remark \ref{Xconnected} every $X \in \G$ of size $2$ induces exactly one edge.
Hence, by the above inequalities, we have
$\sum_{X \in \G} (2|X| - 3)=\sum_{X \in \H} (2|X| - 3)+|E\setminus F|
\geq \sum_{X \in \H} 2|X|-\sum_{X \in \H} |X_B|+3(|V|-|V(\H)|)
=(\sum_{X \in \H} 2|X_I| +\sum_{X \in \H} |X_B|-2|V(\H)|) + (|V|-|V(\H)|)+2|V|
\geq  2|V|.$
\medskip

To see the last inequality, let $v \in V(\H)$. Then $v\in V$ and hence $n-|V(\H)|\geq 0.$ If $v$ belongs to exactly one $X' \in \H$, then $v\in X'_I$; so $v$ contributes $2$ in $\sum_{X \in \H} 2|X_I|$. If $v$ belongs to at least two $X', X'' \in \H$, then $v\in X'_B$ and $v\in X''_B;$ so $v$ contributes at least $2$ in $\sum_{X \in \H} |X_B|$ and hence $\sum_{X \in \H} 2|X_I| +\sum_{X \in \H} |X_B|-2|V(\H)|\geq 0.$

Hence $2|V| - 3 \geq r_{\mathcal{R}}(E) \geq 2|V|$, a contradiction.\end{proof}

\begin{proof}[Proof of Theorem \ref{Ja_Jo}] Note that in the lower bound on $|E\setminus F|$, $d_{G-X_B}(X_I)\geq 6-|X_B|$ can be replaced by $d_{G-X_B}(X_I)\geq 6-2|X_B|$, and the same proof works. This means that instead of $6$-connectivity, we used in fact $(6,2)$-connectivity. 
\end{proof}

\begin{proof}[Proof of Theorem \ref{main_th}] Suppose that there exist integers $k,\ell$ and a graph $G=(V,E)$ contradicting the theorem. We use the matroid $\mathcal{M}_{k,\ell}$ defined above.
 Choose $F$ a smallest-size set of edges that minimizes the right hand side of (\ref{rs}). By (\ref{rr}), we can define $\H$ a collection of subsets of $V$ such that $\{E(X) \cap F, X \in \H\}$ partitions $F$ and $r_{\mathcal{R}}(F) = \sum_{X \in \H} (2|X| - 3).$
Since $G$ is a counterexample and by (\ref{rr}) and (\ref{rc}),
\begin{equation} \label{main_ineq}
 k(2n-3) + \ell(n-1) > r_{\mathcal{M}_{k,\ell}}(E) = k\sum_{X \in \H} (2|X| - 3) +\ell(n-c(G_F)) + |E \setminus  F|.
\end{equation}
By $k \geq 1$, $G$ is connected, thus, by (\ref{main_ineq}), $V \notin \H$. 
Recall the notations, for $X \in \H$,  $X_B = X \cap (\cup_{Y \in \H-X} Y)$ and $X_I = X \setminus X_B$ and
the definition $\H' = \{X \in \H \textrm{ : } X_I \neq \emptyset\}$.
Denote $\K$ the set of connected components of $G_F$ intersecting no set of $\H'$.
By Remark \ref{Xconnected}, for $X \in \H'$, $X$ induces a connected subgraph of $G_F$, thus a connected component of $G_F$ intersecting $X \in \H'$ contains $X$ and is the only connected component of $G_F$ containing $X$. So by definition of $\K$,
\begin{equation} \label{connect3}
|\H'| \geq c(G_F) - |\K| . 
\end{equation}

Let us first show a lower bound on $|E \setminus F|.$

\begin{claim} \label{otheredges}
 $|E \setminus F|\geq k \bigg(3|\H'| -\sum_{X \in \H'} |X_B| +3|\K|\bigg) + \ell c(G_F).$
\end{claim}
\begin{proof} For $X \in \H'$, $X_I \neq \emptyset$ and $X_I \cup X_B = X \neq V$. Thus by $(6k+2\ell,2k)$-connectivity of $G$,
for $X \in \H'$ and for $K \in \K$, 
\begin{eqnarray}
d_{G-X_B}(X_I) &\geq &(6k+2\ell) - 2k|X_B|, \label{connect1}\\
d_G(K) &\geq &6k+2\ell. \label{connect2}
\end{eqnarray}
Since every edge of $F$ is induced by an element of $\H$ and by definition of $X_I$, for $X \in \H'$, no edge of $F$ contributes to $d_{G-X_B}(X_I)$. Each $K \in \K$ is a connected component of the graph $G_F$, thus no edge of $F$ contributes to $d_G(K)$.
Hence, by (\ref{connect1}), (\ref{connect2}), (\ref{connect3}) and $\ell \geq 0$, we obtain the required lower bound on $|E \setminus F|$,
\begin{eqnarray*}
 |E \setminus F| & \geq & \frac{1}{2}\bigg( \sum_{X \in \H'} d_{G-X_B}(X_I) + \sum_{K \in \K}d_G(K)\bigg) \\
    {} & \geq & \frac{1}{2}\bigg((6k+2\ell)|\H'| - 2k\sum_{X \in \H'} |X_B| + (6k+2\ell)|\K|\bigg) \\
    {} & \geq & k \bigg(3|\H'| -\sum_{X \in \H'} |X_B| +3|\K|\bigg) + \ell\bigg(|\H'| + |\K|\bigg)\\
    {} & \geq & k \bigg(3|\H'| -\sum_{X \in \H'} |X_B| +3|\K|\bigg) + \ell c(G_F).
\end{eqnarray*}
\vspace{-1.3cm}\[\qedhere\]
\end{proof}
\medskip

\begin{claim} \label{3c}
  $\sum_{X \in \H \setminus \H'}|X_B| \geq 3(|\H| - |\H'|)$.
\end{claim}
\begin{proof}
 By definition of $\H'$, $X_B=X$ for all $X \in \H \setminus \H'$. So to prove the claim it suffices to show that every $X \in \H$ satisfies $|X| \geq 3$.
 Suppose there exists $Y \in \H$ such that $|Y|=2$. By Remark \ref{Xconnected} and since $G$ is simple, $Y$ induces exactly one edge $e$. Define $F'' = F - e$ and $\H'' = \H - Y$. Note that $\{E(X)\cap F'', X \in \H''\}$ partitions $F''$, hence by (\ref{rr}) and the choice of  $\H$, 
  \begin{equation} \label{rin}
 r_{\mathcal{R}}(F'') \leq \sum_{X \in \H''} (2|X| - 3)= r_{\mathcal{R}}(F)-(2|Y|-3) = r_{\mathcal{R}}(F)-1.
 \end{equation}
  Note also that $c(G_{F''}) \geq c(G_F)$, thus by (\ref{rc}) and $\ell \geq 0$, 
 \begin{equation} \label{lin}
 \ell r_{\mathcal{C}}(F'') \leq \ell r_{\mathcal{C}}(F).
 \end{equation}
Since $|F''| < |F|$, the choice of $F$ implies that $F''$ doesn't minimizes the right hand side of (\ref{rs}). Hence by (\ref{rin}), (\ref{lin}), the definition of $F''$, $|Y|=2$, and $k \geq 1$, we have the following contradiction:
\begin{eqnarray*}
0 &<& \bigg(kr_{\mathcal{R}}(F'') +\ell r_{\mathcal{C}}(F'') + |E \setminus F''|\bigg) - \bigg(kr_{\mathcal{R}}(F) +\ell r_{\mathcal{C}}(F) + |E \setminus F|\bigg) \\
&=&k\bigg(r_{\mathcal{R}}(F'')-r_{\mathcal{R}}(F)\bigg) +\ell \bigg(r_{\mathcal{C}}(F'')-r_{\mathcal{C}}(F)\bigg)+ \bigg(|E \setminus F''| -  |E \setminus F|)\bigg) \\
&\leq &-k + 0 + |\{e\}| \\
&\leq &0.
\end{eqnarray*}
\vspace{-1.3cm}\[\qedhere\]
\end{proof}\medskip

To finish the proof we show the following inequality with a simple counting argument.

\begin{claim} \label{counting}
 $2|\K| + \sum_{X \in \H} 2|X_I| +\sum_{X \in \H} |X_B| \geq 2n$.
\end{claim}
\begin{proof} Let $v \in V$. If $v$ belongs to no $X \in \H$, then $\{v\} \in \K$ and $v$ contributes $2$ in $2|\K|$. If $v$ belongs to exactly one $X' \in \H$, then $v\in X'_I$ and $v$ contributes $2$ in $\sum_{X \in \H} 2|X_I|$. If $v$ belongs to at least two $X', X'' \in \H$, then $v\in X'_B, v\in X''_B$ and $v$ contributes at least $2$ in $\sum_{X \in \H} |X_B|$. The claim follows. 
\end{proof}

Thus we get, by Claims \ref{otheredges}, \ref{3c} and \ref{counting},
\begin{align*}
 &k\sum_{X \in \H} (2|X| - 3) + |E \setminus F|+\ell (n-c(G_F))\\
    &\geq k \sum_{X \in \H} 2|X| - 3k|\H| + k \bigg(3|\H'| -\sum_{X \in \H'} |X_B| +3|\K|\bigg) + \ell c(G_F) + \ell (n -c(G_F)) \\
    &\geq  k \bigg(\sum_{X \in \H} 2|X| -3|\H| + 3|\H'| -\sum_{X \in \H'} |X_B| +3|\K|\bigg) + \ell n\\
    & \geq  k \bigg(\sum_{X \in \H} 2|X| -\sum_{X \in \H} |X_B| +2|\K|\bigg) + \ell n \\
    & \geq  k \bigg(2|\K| + \sum_{X \in \H} 2|X_I| +\sum_{X \in \H} |X_B|\bigg) + \ell n\\
    & \geq  2kn + \ell n.
\end{align*}
By $k \geq 1$ and $\ell \geq 0$, this contradicts (\ref{main_ineq}). \end{proof}

Remark that the proof actually shows that if $G$ is simple and $(6k+2\ell,2k)$-connected and if $F \subseteq E$ is such that $|F| \leq 3k + \ell$, then $G'=(V,E \setminus F)$ contains $k$ rigid spanning subgraphs and $\ell$ spanning trees pairwise edge disjoint.\\

\bibliographystyle{plain}

\begin{thebibliography}{10}

\bibitem{Berg_Jordan}
A.~R. Berg and T.~Jordán.
\newblock Two-connected orientations of eulerian graphs.
\newblock {\em Journal of Graph Theory}, 52(3):230--242, 2006.

\bibitem{Edmonds_1968}
J.~Edmonds.
\newblock Matroid partition.
\newblock In {\em Mathematics of the Decision Science Part 1}, volume~11, pages
  335--345. AMS, Providence, RI, 1968.

\bibitem{Frank_1995}
A.~Frank.
\newblock Connectivity and network flows.
\newblock In {\em Handbook of combinatorics}, pages 117--177. Elsevier,
  Amsterdam, 1995.

\bibitem{Jackson_Jordan2009}
B.~Jackson and T.~Jord\'an.
\newblock A sufficient connectivity condition for generic rigidity in the
  plane.
\newblock {\em Discrete Applied Mathematics}, 157(8):1965--1968, 2009.

\bibitem{Jackson_Jordan2005}
B.~Jackson and T.~Jordán.
\newblock Connected rigidity matroids and unique realizations of graphs.
\newblock {\em Journal of Combinatorial Theory, Series B}, 94(1):1 -- 29, 2005.

\bibitem{Jordan_2005}
T.~Jord{\'a}n.
\newblock On the existence of k edge-disjoint 2-connected spanning subgraphs.
\newblock {\em Journal of Combinatorial Theory, Series B}, 95(2):257--262,
  2005.

\bibitem{Kaneko_Ota2000}
A.~Kaneko and K.~Ota.
\newblock On minimally $(n,\lambda)$-connected graphs.
\newblock {\em Journal of Combinatorial Theory, Series B}, 80(1):156 -- 171,
  2000.

\bibitem{KiraSzig}
Z.~Kir\'{a}ly and Z.~Szigeti.
\newblock Simultaneous well-balanced orientations of graphs.
\newblock {\em J. Comb. Theory Ser. B}, 96(5):684--692, 2006.

\bibitem{Lovasz_Yemini}
L.~Lovász and Y.~Yemini.
\newblock On generic rigidity in the plane.
\newblock {\em J. Algebraic Discrete Methods}, 3(1):91--98, 1982.

\bibitem{Thomassen_1989}
C.~Thomassen.
\newblock Configurations in graphs of large minimum degree, connectivity, or
  chromatic number.
\newblock {\em Annals of the New York Academy of Sciences}, 555(1):402--412,
  1989.

\end{thebibliography}

\end{document}